%% file: main.tex
\documentclass[letterpaper,10pt,conference]{ieeeconf}

\usepackage{cite}
\usepackage{graphicx}
\usepackage{tikz}
\usepackage[cmex10]{amsmath}

\usepackage{xurl}                 
\usepackage[hidelinks]{hyperref}
\usepackage{algpseudocode}
\usepackage{algorithmicx} 
\usepackage{array}
\usepackage{multirow}

\usepackage{epsfig}
\usepackage{amsfonts,amssymb,multirow,bigstrut,booktabs,tabularx, latexsym}

\usepackage[mathscr]{eucal}
\usepackage{verbatim}

\usepackage{amsthm}
\usepackage{algorithm}

\usepackage{stmaryrd}

\IEEEoverridecommandlockouts  
\newtheorem{definition}{Definition}
\newtheorem{lemma}{Lemma}

\newtheorem{theorem}{Theorem}

\newtheorem{assumption}{Assumption}

\newtheorem{problem}{Problem}

\title{\LARGE \bf A Scalable Design Approach to Resilient Architectures for Interconnected Cyber-Physical Systems: Safety Guarantees under Multiple Attacks}

\author{Eman Badr$^{1}$ and Abdullah Al Maruf$^{2}$
     \thanks{$^1$California State University, Los Angeles (Cal State LA).
        {\tt\small ebadr@calstatela.edu}}
     \thanks{$^2$California State University, Los Angeles (Cal State LA).
        {\tt\small amaruf@calstatela.edu}}	
        \thanks{The authors thank ChatGPT for improving the clarity and conciseness of a few sentences in this paper.}
}

\begin{document}

\maketitle

\begin{abstract}

Complex, interconnected cyber-physical systems (CPS) are increasingly prevalent in domains such as power systems. Cyber-resilient architectures have been proposed to recover compromised cyber components of CPS. Recent works have studied tuning the recovery times of such architectures to guarantee safety in single-system settings. Extending these designs to interconnected CPS is more challenging, since solutions must account for attacks on multiple subsystems that can occur in any order and potentially infinite possible temporal overlap. This paper aims to address the aforementioned challenge by developing a scalable framework to assign resilient architectures and to inform the tuning of their recovery times. Our approach introduces a scalar index that quantifies the impact of each subsystem on safety under compromised input. These indices aggregate linearly across subsystems, enabling scalable analysis under arbitrary attack orderings and temporal overlaps. We establish a linear inequality relating each subsystem’s index and recovery time that guarantees safety and guides resilient architecture assignment. We also propose a segmentation-based approach to strengthen the previously derived conditions. We then present algorithms to compute the proposed indices and to find a cost-optimal architecture assignment with a safety guarantee. We validate the framework through a case study on temperature regulation in interconnected rooms under different attack scenarios.


\end{abstract}
\input{introduction_v8}

\input{formulation3}

\input{result4}

\input{simulation}

\input{conclusion}

\section*{Acknowledgment}
We thank the Network Security Lab (NSL) at the University of Washington. In particular, we thank Prof. Radha Poovendran and Prof. Luyao Niu (University of Washington), Prof. Andrew Clark (Washington University in St. Louis), Prof. Bhaskar Ramasubramanian (Western Washington University), and Prof. Sukarno Mertoguno (Georgia Institute of Technology) for their valuable suggestions on this work.

\bibliographystyle{IEEEtran}
\bibliography{MyBib}

\end{document}

%% file: introduction_v8.tex
\section{Introduction}\label{sec:intro}

Cyber-physical systems (CPSs) are exposed to malicious cyberattacks, as reported across numerous applications, including transportation \cite{greenberg2019hackers} and power systems~\cite{case2016analysis}. Such attacks may cause safety violations, with consequences ranging from equipment damage to significant risks to human safety. To this end, safety verification, safety-critical control design, and cyber-resilient architectures have been extensively studied~\cite{abdirash2025decentralized,ames2016cbf,alhidaifi2024survey,mo2009secure,cohen2020approximate,mertoguno2019physics,romagnoli2019design}.

Research on cyber-resilient architectures focuses on mechanisms such as redundancy, software diversity, and reactive restarts to enable recovery after compromise~\cite{mertoguno2019physics,arroyo2019yolo, abdi2018guaranteed,romagnoli2019design}. Recent works have examined tuning recovery time and the control policy to guarantee safety under selected resilient architectures~\cite{niu2022restart, maruf2022timing, niu2022analytical, abdi2018guaranteed}. However, these studies largely consider single-system settings and are not readily applicable to large interconnected systems with nonlinear dynamics. Due to coupling, the effects of an attack on one subsystem can propagate to others, affecting their operation. Moreover, in an interconnected CPS, multiple subsystems may be attacked in a coordinated fashion, in arbitrary orders, with potentially infinitely many possibilities of temporal overlap \cite{ten2017impact}. This combinatorial nature poses a scalability challenge as the number of subsystems grows. To the best of our knowledge, no scalable method has yet been developed in the literature for assigning cyber-resilient architectures across multiple subsystems in an interconnected CPS with formal safety guarantees.

This paper bridges that gap by developing a scalable framework that handles multiple actuation attacks on subsystems in arbitrary order with arbitrary temporal overlap. Our solution introduces a scalar metric that relates a resilient architecture’s recovery capability to safety-critical assignment decisions. Similar compositional approaches have been used to decompose safety constraints across subsystems or to design centralized controllers for safety \cite{maruf2022compositional,nejati2020cbc,yang2020smallgain, niu2023compositional}. However, those results do not apply here, as our setting requires new indices and safety conditions tailored to resilient-architecture assignment. We make the following contributions:

\begin{itemize}
\item We introduce a \emph{Criticality Index} (CI) that quantifies the impact of a subsystem’s compromised input on system's safety. The proposed CIs aggregate linearly across subsystems, enabling scalable analysis under arbitrary attack orderings and temporal overlaps.


\item We derive a sufficient condition for safety, expressed as a linear inequality that relates each subsystem’s CI to its recovery time, thereby guiding assignment of resilient architectures and tuning of their recovery times.

\item We further refine the sufficient condition via a segmentation-based approach that permits employing architectures with longer recovery times.

\item We formulate a Sum-of-Squares (SOS) optimization problem to compute the criticality indices and present an algorithm to find a cost-optimal architecture assignment with a safety guarantee.

\item We validate our proposed framework through a detailed case study on temperature regulation in interconnected rooms for different attack scenarios.
\end{itemize}

This paper is organized as follows. Section \ref{sec:formulation} defines the system model and problem formulation. Section \ref{sec:results} presents our main results, and Section \ref{sec:simulation} showcases a case study. Section \ref{sec:conclusion} concludes the paper.

%% file: formulation3.tex
\section{System Model and Problem Formulation}\label{sec:formulation}
We consider an interconnected Cyber-Physical System (CPS) composed of a set of $N$ subsystems, denoted as $\mathcal{S}=\{S_1, S_2, \dots, S_N\}$.  Each subsystem $S_i$ evolves according to the following dynamics:
\begin{equation}
    S_i: \dot{x}_i = f_i(x) + g_i(x)~ u_i \label{eq:system}
\end{equation}
\noindent In \eqref{eq:system}, $x_i \in \mathbb{R}^{n_i}$ is the state of subsystem $S_i$, and $x = [x_1^T~ x_2 ^T \cdots x_N^T]^T \in \mathbb{R}^n$ is the overall state of the system, where $n = \sum_{i=1}^N n_i$. The control input applied to subsystem $S_i$ is $u_i \in \mathbb{R}^{r_i}$. The functions $f_i : \mathbb{R}^n \to \mathbb{R}^{n_i}$ and $g_i : \mathbb{R}^n \to \mathbb{R}^{n_i \times r_i}$ are assumed to be Lipschitz continuous. The nonlinear dynamics in (\ref{eq:system}) represent an interconnected system coupled via the state $x$. We assume that the control input to each subsystem $S_i$ is bounded, i.e., $u_i \in \mathcal{U}_i$, where $\mathcal{U}_i = \prod_{j=1}^{r_i} [\underline{u}_{i,j}, \overline{u}_{i,j}]$ with $\underline{u}_{i,j} < \overline{u}_{i,j}$.
Each subsystem $S_i$ employs a nominal control policy defined as a function $\hat{u}_i=\mu_i(x): \mathbb{R}^n \rightarrow \mathcal{U}_i$, which maps the full state of the system to an admissible control input.


We assume there is a safety constraint imposed on the overall system, encoded by a continuously differentiable function $h: \mathbb{R}^n \rightarrow \mathbb{R}$, such that the system state must remain within the set:
\begin{equation}
    \mathcal{C} = \{ x \in \mathbb{R}^n : h(x) \geq 0 \} \nonumber
\end{equation}
This set $\mathcal{C}$, referred to as the \textit{safety set}, characterizes the allowable and safe operating region of the system. 

Each subsystem $S_i$ is vulnerable to cyberattacks by malicious adversaries. Such events may compromise the control input $u_i$, replacing it with the corrupted input $\tilde{u}_i \in \mathcal{U}_i$, which can potentially violate the safety constraint $h(x) \geq 0$. For our analysis, we consider an attack cycle $[t_0,t_f]$, where each subsystem may be compromised once by the attackers, in an arbitrary order and with any possible temporal overlaps. We emphasize that infinitely many attack scenarios may arise due to the arbitrary timing of subsystem compromises. Fig.~\ref{fig:Attack} illustrates several such scenarios for a system with three subsystems. The time $t_f$ designates the termination of an attack cycle, i.e., the instant by which all compromised subsystems have recovered.

\begin{figure}[H]
    \centering
    \includegraphics[scale=0.37]{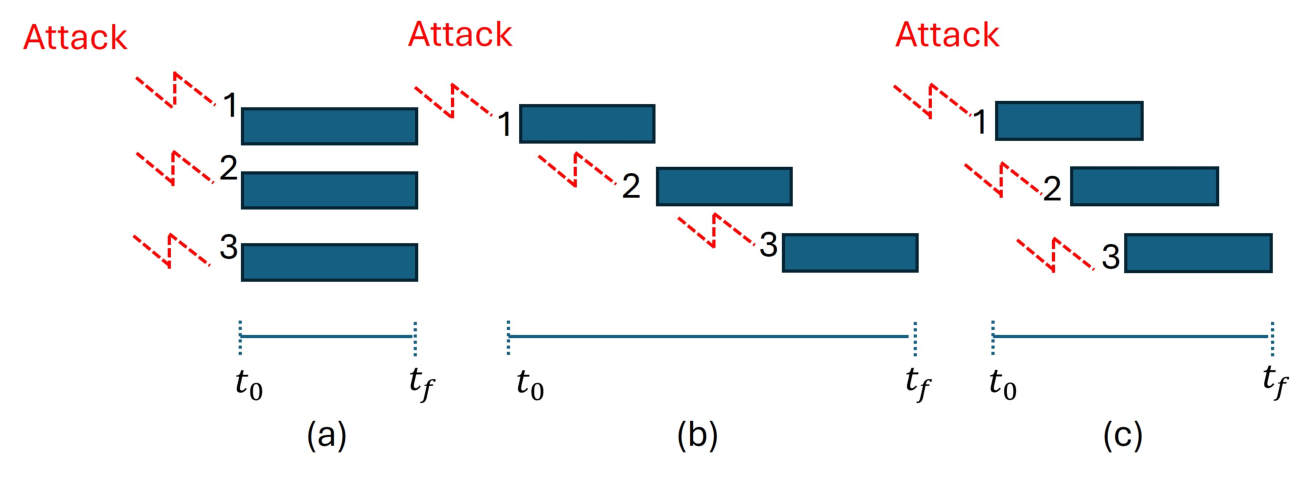}
    \caption[Three different attack scenarios.] {Illustration of three attack scenarios on an interconnected CPS with three subsystems: (a) all subsystems are compromised simultaneously, (b) each subsystem is attacked sequentially without overlap among the compromised subsystems, and (c) subsystems are attacked with 50\% temporal overlap among the compromised subsystems. The red dashed arrows denote attack events, while the blue bars denote the duration of subsystem compromise until recovery. Recovery times are shown equal for simplicity but are not assumed equal in our analysis. Attacks may occur in different order with infinitely many possible temporal overlaps among subsystems.} 
    \label{fig:Attack}
\end{figure}

To maintain system safety against cyberattacks, we assume that each subsystem $S_i$ employs a resilient architecture that reactively recovers the compromised subsystem from a cyberattack within a certain recovery time.\footnote{For example, the Byzantine Fault Tolerant++ (BFT++) architecture has been shown to guarantee the recovery of a compromised controller within certain epochs (or controller operation cycles) \cite{mertoguno2019physics}.} Specifically, we assume a set of \( J \) available resilient architectures \( \mathcal{A}= \{A_1, A_2, \dots, A_J\} \), each characterized by a guaranteed recovery time \( t(A_j) \). We define a mapping $M: \mathcal{S} \rightarrow \mathcal{A}$ where \( M(S_i) = A_j \) indicates  the architecture \( A_j \) is employed in subsystem \( S_i \). Our objective in this study is to determine how to assign the resilient architectures to individual subsystems such that the overall system remains safe under any arbitrary attack scenario and multiple attack cycles. Specifically, we aim to derive conditions for assigning the architecture $M(S_i)$ to each subsystem $S_i$ for $i=1, 2, \dots, N$ that ensure strict safety guarantees. We formally state the problem as follows:

\begin{problem} \label{prob1}
Select a mapping  $M: \mathcal{S} \rightarrow \mathcal{A}$ such that the interconnected system governed by (\ref{eq:system}) remains safe under any cyberattack scenario during an attack cycle. Additionally, determine the conditions under which the mapping $M$ guarantees system safety across multiple attack cycles.
\end{problem}


%% file: result4.tex
\section{Main Result} \label{sec:results}


In this section, we present our main results for the resilient architecture assignment with safety guarantees. We begin by introducing the criticality index (CI), which forms the foundation of our scalable solution approach. Next, we derive a sufficient condition for safety in the form of a linear inequality that links each subsystem’s CI with its recovery time. We then propose a segmentation-based strategy to tighten this condition. Finally, we formulate the computation of the CI as a sum-of-squares (SOS) optimization problem and provide an algorithm to find a cost-optimal architecture assignment with guaranteed safety.

\subsection{CI-Based Safety Guarantee} \label{sec:Res} 
We aim to develop a scalable solution to the resilient architecture assignment problem that applies to arbitrary attack scenarios in a large interconnected CPS, regardless of the order or temporal overlap of subsystem compromises. 
Before presenting our solution approach, we first make a few remarks about the nominal policy and the system’s operating range.  

To guarantee safety against cyberattacks, it is necessary that the system does not operate on the boundary of the safety set, i.e., $\{x:h(x)=0\}$. 
Accordingly, we consider that the system maintains a margin from the boundary of the safety set, characterized by a parameter $c>0$. 
Specifically, prior to the occurrence of an attack, the system operates within $\mathcal{C}_c=\{x:h(x)\geq c\}$, i.e., $x(t_0)\in \mathcal{C}_c$. Additionally, in the absence of cyberattacks, the nominal policy keeps the system within $\mathcal{C}_c$. Formally, we make the following assumptions for the nominal control policy $\hat{u}_i=\mu_i(x)$ applied to each subsystem $S_i$, for $i=1,2,\dots, N$.

\begin{assumption}
\begin{equation} 
\sum_{i=1}^N \frac{\partial h}{\partial x_i} \big(f_i(x)+g_i(x)\hat{u}_i(x)\big) \;\geq\; \alpha(h(x)-c), 
\quad \forall x\in \mathcal{C}_c 
\end{equation} \label{assump1}
\begin{equation} 
\sum_{i=1}^N \frac{\partial h}{\partial x_i} \big(f_i(x)+g_i(x)\hat{u}_i(x)\big) \;\geq\; \frac{c}{\tau}, 
\quad ~~~~\forall x\in \mathcal{C}\setminus\mathcal{C}_c 
\end{equation}
\end{assumption}
\hfill\(\Box\)

Here, $\alpha(\cdot)$ denotes an extended class-$\mathcal{K}$ function. 
The first condition ensures that the set $\mathcal{C}_c$ is forward invariant, i.e., $x(t)\in\mathcal{C}_c$, under the nominal policy in the absence of attacks. 
This follows directly from the properties of control barrier function (CBF)~\cite{ames2016cbf}. 
The second condition ensures that if the system state lies in $\mathcal{C}\setminus\mathcal{C}_c$, the nominal policy drives the state back into $\mathcal{C}_c$ within a finite time $\tau$.  
The first condition is practical, since a nominal policy should not itself cause a safety violation. 
The second condition is necessary to guarantee safety across multiple attack cycles, as discussed later in this section.\footnote{One could replace the second condition with a condition of finite-time convergence control barrier function (FCBF), which guarantees system's return to $\mathcal{C}_c$ within $t\in\left[0,\frac{c^{(1-p)}}{\gamma(1-p)}\right]$ for some $\gamma>0$ and $p\in[0,1)$~\cite{li2018formally}. Our results still hold under this replacement by taking $\tau=\frac{c^{1-p}}{\gamma(1-p)}$ throughout the subsequent analysis.}

We now present our solution approach, which relies on constructing a scalar metric for each subsystem $S_i$ that quantifies its potential impact on overall system's safety in the event of a compromise. 
We refer to this metric as the \textit{Criticality Index} (CI), formally defined below.


\begin{definition}\label{def:rho}
For each subsystem $S_i$, we define the criticality index (CI), denoted as $\hat{\rho}_i$, with respect to the set $\mathcal{C} \setminus \mathcal{C}_c$ as
\begin{equation}
\hat{\rho}_i = \inf_{x \in \mathcal{C} \setminus \mathcal{C}_c,\, u_i \in \mathcal{U}_i} \left\{ \frac{\partial h}{\partial x_i} \, g_i(x)(u_i - \hat{u}_i) \right\}
\end{equation}
\end{definition}

The index $\hat{\rho}_i$ is a negative scalar that quantifies the worst-case rate of degradation of the safety function $h(x)$ resulting from deviations in the control input from the nominal policy in subsystem $\mathcal{S}_i$. A more negative value of $\hat{\rho}_i$ indicates a more critical subsystem, in the sense that its compromised input leads to faster safety degradation. We emphasize that the criticality index is defined over the set $\mathcal{C} \setminus \mathcal{C}_c$, rather than the entire safety set $\mathcal{C}$, as this yields a less conservative result.

When \( \hat{\rho}_i \) is difficult to compute, we may estimate it by finding a scalar \( \rho_i \in \mathbb{R} \) such that for all \( x \in \mathcal{C} \setminus \mathcal{C}_c \) and \( u_i \in \mathcal{U}_i \)
\begin{equation} \label{eq:rho}
\frac{\partial h}{\partial x_i} g_i(x) \left( u_i(x) - \hat{u}_i(x) \right) \geq \rho_i 
\end{equation}
By Definition \ref{def:rho}, $\hat{\rho}_i\geq \rho_i$ for any $\rho_i$ satisfying \eqref{eq:rho}. 

We now derive a sufficient condition to guarantee the system's safety that links the recovery time of each subsystem's employed resilient architecture to its CI. Specifically, it describes what should be the recovery time of the architecture $M(S_i)$ implemented in each subsystem $S_i$ so that the system remains safe in the occurrence of cyberattacks in multiple subsystems. The result is formalized below.

\begin{theorem} \label{Th1} 
Consider an attack cycle $[t_0,\, t_f]$ where each subsystem is compromised at most once in the cycle. Assume $x(t_0) \in \mathcal{C}_c$ and the nominal policy satisfies Assumption~1. Suppose each subsystem $S_i$ employs a cyber-resilient architecture $M(S_i)$ according to the mapping $M$. Then the system is guaranteed to be safe for all $t \in [t_0,t_f]$ if  
\begin{equation} \label{eq:Safety Condition_1}
\sum_{i=1}^N \rho_i~t(M(S_i)) + c\geq0 
\end{equation}
\end{theorem} 
\begin{proof} Assuming $x(t_0)\in \mathcal{C}_c$, we can write:

\begin{align}\label{prf1}
\resizebox{0.95\linewidth}{!}{$
\begin{aligned}
h(x(t)) &= \int_{t_0}^{t} \dot{h}(x(t))\, dt + h(x(t_0)) \\
&\geq \int_{t_0}^{t} \Bigg(\sum_{i=1}^N 
    \frac{\partial h}{\partial x_i}\big(f_i(x)+g_i(x)~u_i\big)\Bigg)\,dt + c \\
&= \int_{t_0}^{t} \Bigg(\sum_{i=1}^N 
    \frac{\partial h}{\partial x_i}\big(f_i(x)+g_i(x)(u_i-\hat{u}_i)+g_i(x)\hat{u}_i\big)\Bigg)\,dt + c \\
&= \int_{t_0}^{t} \Bigg(\sum_{i=1}^N 
    \frac{\partial h}{\partial x_i}\big(f_i(x)+g_i(x)~\hat{u}_i\big)\Bigg)\,dt \\
&\quad + \sum_{i=1}^N \Bigg(\int_{t_0}^{t} 
    \frac{\partial h}{\partial x_i} g_i(x)(u_i-\hat{u}_i)\,dt\Bigg) + c
\end{aligned}
$}
\end{align}

First, note that the system remains safe if it stays within the set $\mathcal{C}_c$ and it must pass through the set $\mathcal{C}\setminus\mathcal{C}_c$ in order to become unsafe. Therefore, to derive a sufficient condition for safety, we consider the worst-case scenario in which the system remains in the set $\{x \in \mathbb{R}^n : h(x) \leq c \}$ during the entire compromise period. Now, by Assumption~1, we have $\sum_{i=1}^N \frac{\partial h}{\partial x_i} (f_i(x)+g_i(x)~\hat{u}_i)\geq 0$ for $x\in \mathcal{C} \setminus \mathcal{C}_c$. Furthermore, for each subsystem $S_i$, we have $u_i\neq \hat{u}_i$ only for $t(M(S_i))$ time-duration, as the maximum time an attacker can compromise a subsystem in a single attack cycle is bounded by its recovery time. Utilizing this fact together with \eqref{eq:rho}, where $\rho_i$ is a negative scalar, we can write for each subsystem in a single attack cycle that 
\begin{equation} \label{eq:deg_amnt}
 \int_{t_0}^{t} \frac{\partial h}{\partial x_i}g_i(x)({u}_i-\hat{u}_i)~dt \geq  \int_{0}^{t(M(S_i))} \rho_i ~dt \geq \rho_i~t(M(S_i))   
\end{equation} 
Substituting these into \eqref{prf1}, for $t\in[t_0,t_f]$ we get
\begin{align}  
h(x(t))\geq \sum_{i=1}^N \rho_i~t(M(S_i)) +c
\end{align} 
Therefore, if \eqref{eq:Safety Condition_1} is satisfied, then $h(x(t))\geq0$, i.e. $x(t)\in \mathcal{C}$, which implies the system is safe for all $t \in [t_0,t_f]$.

\end{proof}

Theorem 1 establishes a simple linear inequality that can be easily evaluated to guarantee that the assigned architectures keep the system safe throughout any single attack cycle. Moreover, the theorem suggests a simple strategy for architecture assignment—subsystems with more negative CIs should be assigned architectures with shorter recovery times. In addition to verifying safety, the condition in Theorem 1 also provides guidance for designing the recovery time of each individual subsystem. Let $T_i$ be the tunable recovery time of subsystem $S_i$. Then, based on the proof of Theorem 1, the system will maintain safety if $\sum_{i=1}^N \rho_i~T_i + c>0$. It is easy to check that a necessary condition for safety under Theorem \ref{Th1} is: $T_i \leq -\frac{c}{\rho_i}$ for all $i=1, \cdots, N$. It is also worth noting that a larger value of $c$ permits the use of architectures with longer recovery times. This is intuitive, as a higher $c$ implies the system operates farther from the safety boundary, allowing more time to recover for each subsystem without violating safety constraints.

We now present the safety conditions for the case of multiple attack cycles in Theorem~\ref{Th2}.

\begin{theorem} \label{Th2} 
Consider that the system experiences multiple attack cycles: $[t_{0,1},~ t_{f,1}], [t_{0,2},~ t_{f,2}], \cdots$, where $t_{0,1}=t_0$ and each subsystem is compromised at most once per attack cycle. Assume $x(t_{0}) \in \mathcal{C}_c$ and the nominal policy satisfies Assumption \ref{assump1}. Suppose each subsystem $S_i$ employs a cyber-resilient architecture $M(S_i)$ according to the mapping $M$,  which satisfies the condition given by \eqref{eq:Safety Condition_1}. Then, the system is guaranteed to remain safe for all $t\geq t_0$, provided that the time gap between any two consecutive attack cycles is at least $\tau$; that is, $t_{0,k+1}-t_{f,k}\geq \tau$ for all $k=1, 2, \cdots$.
\end{theorem}

\begin{proof}
After the first attack cycle, all controllers are fully recovered and resume their nominal policies, while maintaining safety as guaranteed by Theorem~1. Hence, we have $h(x(t_{f,1}))\geq 0$ and $u_i=\hat{u}_i$ for $t \in [t_{f,1}, t_{0,2}]$. If the subsequent attack cycle begins after a delay of $\tau$, then according to Assumption~\ref{assump1}, the system will have returned to the set $\mathcal{C}_c$. This is because for $x\in \mathcal{C} \setminus \mathcal{C}_c$
\begin{align}
h(x(t)) &= \int_{t_{f,1}}^{t} \dot{h}(x(t))\, dt + h(x(t_{f,1})) \nonumber \\
&\geq \int_{t_{f,1}}^{t} \sum_{i=1}^N \frac{\partial h}{\partial x_i} (f_i(x)+g_i(x)~\hat{u}_i(x))\, dt \nonumber \\
&\geq \frac{c}{\tau} (t-t_{f,1}) \label{mult_att}
\end{align}

Therefore, if $t\geq\tau$, we have $h(x(t))\geq c$ according to \eqref{mult_att}, recalling that the set $\mathcal{C}_c$ is forward invariant under the nominal policy by the Assumption~\ref{assump1}. Therefore, if $t_{0,2}-t_{f,1}\geq \tau$ then at the onset of the next attack cycle, we have $x(t_{0,2}) \in \mathcal{C}_c$ which ensures the system's safety during this attack cycle by Theorem~1. This reasoning can be applied iteratively, thereby guaranteeing the system's safety for all future attack cycles.
\end{proof}

Theorem~\ref{Th2} underscores the importance of Assumption~\ref{assump1}. To ensure safety against multiple attack cycles, the system must return to the initial set $\mathcal{C}_c$ at the beginning of each new attack cycle. The second condition guarantees that the nominal control policy can drive the system from $\mathcal{C} \setminus \mathcal{C}_c$ to $\mathcal{C}_c$ within the finite time $\tau$. A smaller value of $\tau$ enables the system to withstand more frequent attacks, as dictated by Theorem~\ref{Th2}. However, the need to meet mission objectives and satisfy $u_i \in \mathcal{U}_i$ can restrict the attainable value of $\tau$.

\subsection{Improvement of Bound via Segmentation}
In this section, we present an approach to improve the results given in Theorems~\ref{Th1} and \ref{Th2} by tightening the safety condition in \eqref{eq:Safety Condition_1}. The key idea behind our approach is that using a single worst-case rate over the set $\mathcal{C} \setminus \mathcal{C} _c$ can be conservative, especially when $\dot{h}(x)$ varies substantially and the safety margin $c$ is large. To reduce this conservatism, we adopt a segmentation-based strategy in which the set $\mathcal{C} \setminus \mathcal{C}_c$ is partitioned into $K$ equal segments, where $K$ is a positive integer (see the Fig. ~\ref{fig:segmentation}). We then define a separate CI for each segment. As a result, each subsystem $S_i$ is associated with $K$ criticality indices, denoted as $\hat{\rho}_{ij}$ where $j=1, 2, \cdots, K$. The indices for each segment are defined as follows.
\begin{figure}[htbp]
    \centering    \includegraphics[width=0.65\linewidth]{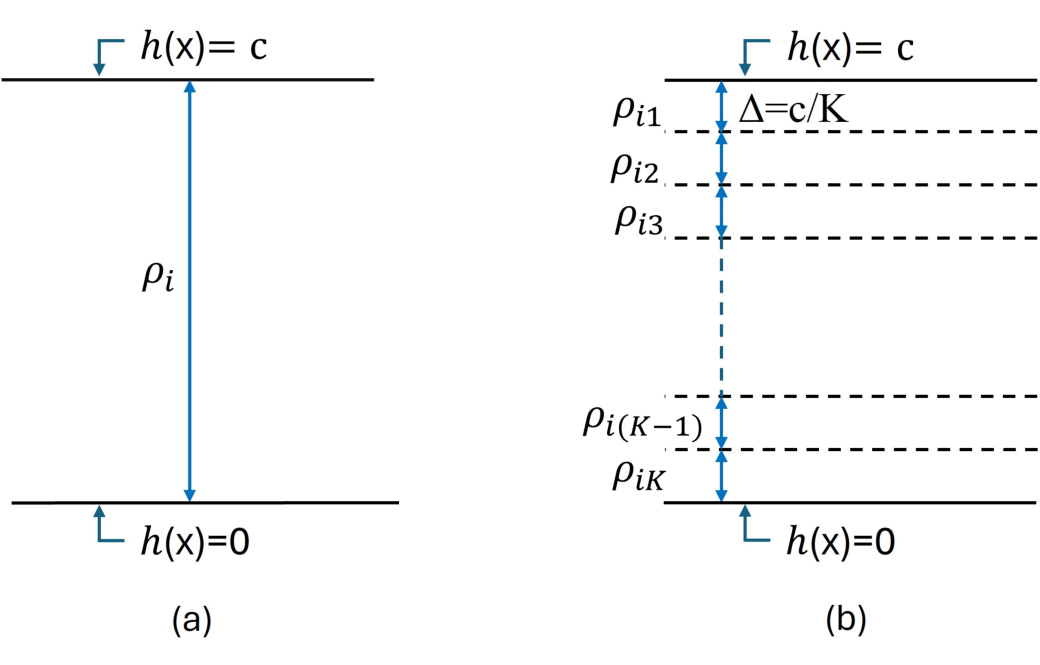} 
    \caption[Safety Degradation Segmentation]{(a) Single criticality index $\rho_i$ representing worst-case degradation rate of the safety function over the set $\mathcal{C} \setminus \mathcal{C}_c$ for the subsystem $S_i$.
    (b) Segmented criticality indices $\rho_{ij}$, where the set $\mathcal{C} \setminus \mathcal{C}_c$ is divided into $K$ equal intervals of size $\Delta = c/K$.}
    \label{fig:segmentation}
\end{figure}
\begin{definition}\label{def:rho_seg}
For each subsystem $S_i$, we define the segmented criticality index (SCI), denoted as $\hat{\rho}_{ij}$, with respect to the safety set $\mathcal{C}_c$ as
\begin{equation} \label{new_rho}
\hat{\rho}_{ij} = \inf_{x \in \mathcal{C}_j,\, u_i \in \mathcal{U}_i} \left\{ \frac{\partial h}{\partial x_i} \, g_i(x)(u_i - \hat{u}_i) \right\}
\end{equation}
where $\mathcal{C}_j=\{x: (j-1)\Delta\leq h(x) \leq j\Delta\}$, $\Delta=\frac{c}{K}$ and $j=1, 2,\cdots, K$.
\end{definition}
Note that since $\mathcal{C}_j \subset \mathcal{C}_c$, therefore $\hat{\rho}_{ij}\geq \hat{\rho}_i$. Similar to before, we can use $\rho_{ij}$ to estimate $\hat{\rho}_{ij}$ by computing a scalar $\rho_{ij}$ such that
\begin{equation} \label{eq:rho_ij}
\frac{\partial h}{\partial x_i} g_i(x) \left( u_i(x) - \hat{u}_i(x) \right) \geq \rho_{ij},~~~\forall x\in \mathcal{C}_j,~\forall u_i\in \mathcal{U}_i
\end{equation}

The usage of SCI allows us to refine the condition \eqref{eq:Safety Condition_1}, resulting in a tighter criterion. In particular, this refinement permits subsystems to employ architectures with longer recovery times while still maintaining the overall safety of the system. To present the updated result, we first arrange the SCI in ascending order. Let $\{\rho'_{i1},\rho'_{i2}, \cdots, \rho'_{iK} \}$ represent a permutation of $\{\rho_{i1},\rho_{i2}, \cdots, \rho_{iK} \}$, ordered such that $\rho'_{i1}\leq \rho'_{i2} \leq \cdots \leq \rho'_{iK} \leq 0$. Then \eqref{eq:Safety Condition_1} can be tightened to
\begin{align}\label{eq: new safety constraint}
\sum_{i=1}^{N} D_i ( M(S_i)) + c  \geq 0
\end{align}
where $D_i(\cdot)$ is given by
\begin{equation} \label{segment}
\resizebox{\linewidth}{!}{$
D_i(T) =
\begin{cases}
\rho'_{i1} T, & 0 \leq T \leq -\Delta \cdot \frac{1}{\rho'_{i1}} \\[15pt]
-\Delta + \rho'_{i2} \left( T + \Delta \cdot \frac{1}{\rho'_{i1}} \right), & -\Delta \cdot \frac{1}{\rho'_{i1}} \leq T \leq -\Delta \cdot \left( \frac{1}{\rho'_{i1}} + \frac{1}{\rho'_{i2}} \right) \\[15pt]
-2\Delta + \rho'_{i3} \left( T + \Delta \left( \frac{1}{\rho'_{i1}} + \frac{1}{\rho'_{i2}} \right) \right), & -\Delta \left( \frac{1}{\rho'_{i1}} + \frac{1}{\rho'_{i2}} \right) \leq T \leq -\Delta \left( \frac{1}{\rho'_{i1}} + \frac{1}{\rho'_{i2}} + \frac{1}{\rho'_{i3}} \right) \\[10pt]
\quad \vdots & \quad \vdots \\[10pt]
-(K - 1)\Delta + \rho'_{iK} \left( T + \Delta \sum_{j=1}^{K-1} \frac{1}{\rho'_{ij}} \right), & T \in \left( -\Delta  \sum_{j=1}^{K-1} \frac{1}{\rho'_{ij}},~ \; -\Delta \sum_{j=1}^{K} \frac{1}{\rho'_{ij}} \right) \\
\end{cases}
$} 
\end{equation}
In the expression above, $D_i(T)$ measures the degradation of the $h(x)$ when the subsystem $S_i$ remains compromised for a duration $T$.  Note that \eqref{eq:Safety Condition_1} is recovered by taking $K=1$. With this updated inequality, we now revise the statements of Theorems 1 and 2 accordingly.

\begin{theorem} \label{Th3} 
Consider that the system experiences an attack cycle $[t_0,~ t_f]$ where each subsystem is compromised at most once in the cycle. Assume $x(t_0) \in \mathcal {C}_c$ and the nominal policy satisfies Assumption~1. Suppose each subsystem $S_i$ employs a cyber-resilient architecture $M(S_i)$ according to the mapping $M$. Then the system is guaranteed to be safe for all $t \in [t_0,t_f]$ if the condition given in \eqref{eq: new safety constraint} holds. 
\end{theorem} 

\begin{proof}
The proof follows the same structure as that of Theorem~1, with modifications to account for segmentation in \eqref{eq:deg_amnt}. 
First, note that for each subsystem $S_i$ and each segment $\mathcal{C}_j$, 
the worst-case degradation of the safety function $h(x)$ is bounded below by 
$\max(-\Delta,\, \rho_{ij}\,t_{ij})$, 
where $t_{i,j}$ denotes the time spent in segment $\mathcal{C}_j$ and satisfies 
$t_{i,j} \geq -\frac{\Delta}{\rho_{ij}}$.

The degradation of $h(x)$ over the segments is analogous to computing the total distance traveled along a path that is divided into equal-length intervals: $-\rho_{ij}$ plays the role of \textit{speed} in each interval, while $\Delta$ represents the \textit{interval length}. Thus, $-\rho_{ij}t_{ij}$ corresponds to the distance covered in each segment, which is upper bounded by $\Delta$, and summing across all segments yields the total distance. Moreover, since the ordered sequence of SCIs satisfies 
$\rho'_{i1}\leq \rho'_{i2}\leq \cdots \leq \rho'_{iK}\leq 0$, the worst-case scenario corresponds to traversing the segments in this order. Following this analogy, we can thereby express the inequality in \eqref{eq:deg_amnt} as
\begin{equation} \label{eq:deg_amnt_new}
 \int_{t_0}^{t} \frac{\partial h}{\partial x_i}\, g_i(x)\,(u_i-\hat{u}_i)\, dt \;\;\geq\; D_i(M(S_i))
\end{equation}
where $D_i(\cdot)$ is given by \eqref{segment}. Substituting this bound into \eqref{prf1} completes the proof.
\end{proof}

Similarly, Theorem~\ref{Th2} can be updated as Theorem~\ref{Th4}. 
The proof follows the same logical structure as that of Theorem~\ref{Th2} and is therefore omitted to avoid repetition.  

\begin{theorem} \label{Th4} 
Consider that the system experiences multiple attack cycles: $[t_{0,1},~ t_{f,1}], [t_{0,2},~ t_{f,2}], \cdots$, where $t_{0,1}=t_0$ and each subsystem is compromised at most once per attack cycle. Assume $x(t_{0}) \in \mathcal{C}_c$ and the nominal policy satisfies Assumption~1. Suppose each subsystem $S_i$ employs a cyber-resilient architecture $M(S_i)$ according to the mapping $M$, which satisfies the condition given by \eqref{eq: new safety constraint}. Then, the system is guaranteed to remain safe for all $t\geq t_0$, provided that the time gap between any two consecutive attack cycles is at least $\tau$; that is, $t_{0,k+1}-t_{f,k}\geq \tau$ for all $k=1, 2, \cdots$.
\end{theorem} 

Relative to Theorems~\ref{Th1} and~\ref{Th2}, Theorems~\ref{Th3} and~\ref{Th4} allow each subsystem to employ architectures with longer recovery times, thereby yielding less conservative safety constraints. The benefit of segmentation becomes evident when comparing the necessary recovery times for each subsystem. 
Let the recovery time of subsystem $S_i$ be denoted by $T_i$. 
Then, the necessary condition for $T_i$ to guarantee safety under Theorems~\ref{Th3} and~\ref{Th4} is $T_i <  -\Delta \sum_{j=1}^{K} \frac{1}{\rho'_{ij}}$.
This follows directly from condition~\eqref{eq: new safety constraint} and the definition of $D_i(\cdot)$ in~\eqref{segment}. In contrast, under Theorems~\ref{Th1} and~\ref{Th2}, the necessary condition for $T_i$ to guarantee safety is $T_i < -\frac{c}{\rho_i}$. Using the facts that $0 \geq \rho'_{ij} \geq \rho_i$ and that the harmonic mean of a set of positive numbers does not exceed their arithmetic mean (HM $\leq$ AM) \cite{protter1977analysis}, we obtain: $-\Delta \sum_{j=1}^{K} \frac{1}{\rho'_{ij}}
= -\frac{c}{K} \sum_{j=1}^{K} \frac{1}{\rho'_{ij}}
\;\;\geq\;\; -\frac{c}{\rho_i}$. Moreover, this inequality is strict whenever $\dot{h}(x)$ is not constant in the set $\mathcal{C} \setminus \mathcal{C}_c$. Therefore, the segmentation-based approach permits each subsystem to tolerate longer recovery times compared to the non-segmented case. 

Note that computational complexity increases with the number of segments. However, these computations can be performed offline as discussed later. Furthermore, beyond a certain segmentation level, the marginal improvement in the bound may become negligible.

\subsection{Computation of SCIs and Algorithms for Architecture Assignment} \label{sec:algo}

In this section, we first present a computational method for determining the segmented criticality indices $\rho_{ij}$. The method leverages Sum-of-Squares (SOS) optimization techniques~\cite{papachristodoulou2013sostools}. To apply this technique, we proceed under specific assumptions regarding the system dynamics, the nominal control policy, and the defined safety constraints throughout the remainder of this section.

\begin{assumption}\label{assump:semi-algebraic}
We assume that $g_i(x)$ and the nominal policy $\hat{u}_i(x)$ are polynomials in $x$ for all $i=1,2, \cdots, N$ and $h(x)$ is polynomial in $x$.
\end{assumption}

Based on the above assumption, we now aim to compute $\rho_{ij}$ using SOS optimization. For that, we first show how to translate conditions \eqref{new_rho} into SOS constraints, as given by the following result. In this approach, we maximize $\rho_{ij}$ to make \eqref{eq:rho_ij} reasonably tight.

\begin{lemma} \label{lemma:1} \label{eq:rho SOS}
Suppose Assumption \ref{assump:semi-algebraic} holds and $p_{i,j}(x,u_i)$, $q_{i,j}(x,u_i)$, $w_{i,j}(x,u_i)$ and $v_{i,j}(x,u_i)$ are SOS polynomials, where $i= 1,\ldots,N$ and $j= 1, \ldots ,K$.  If  $\rho_{ij}$ is the solution to the SOS program:
\begin{align} \label{eq:rho SOS 1}
  \max \; & \rho_{ij} \nonumber\\
  \text{s.t. }~&
    \frac{\partial h}{\partial x_i} g_i(x) (u_i - \hat{u}_i(x)) - \rho_{ij} \nonumber\\
  &+ p_{i,j}(x,u_i)\,(h(x)-j\Delta) \nonumber\\
  &- q_{i,j}(x,u_i)\,(h(x) - (j-1)\Delta) \nonumber\\
  &- \sum_{k=1}^{r_i} w_{i,j,k}(x,u_i)(u_{i,k}-\underline{u}_{i,k}) \nonumber\\
  &- \sum_{k=1}^{r_i} v_{i,j,k}(x,u_i)(\overline{u}_{i,k}-u_{i,k})
    \;\; \text{is SOS}
\end{align}
then $\rho_{ij}$ satisfies \eqref{eq:rho_ij}. Furthermore, $\rho_{ij}=\hat{\rho}_{ij}$ when the
expression in \eqref{eq:rho SOS 1} is quadratic.
\end{lemma}

\begin{proof}
Let \(\rho_{ij}\) be the solution to the SOS optimization problem defined in~\eqref{eq:rho SOS 1}. Since \(p_{i,j}(x,u_i)\), \(q_{i,j}(x,u_i)\), \(w_{i,j}(x,u_i)\), and \(v_{i,j}(x,u_i)\) are SOS polynomials, they are nonnegative for all \(x\) and \(u_i\). Within the domain \(\mathcal{C}_j = \{x : (j-1)\Delta \leq h(x) \leq j\Delta\}\), the inequalities \(h(x)-j\Delta \leq 0\) and \(h(x)-(j-1)\Delta \geq 0\) hold by Definition~\ref{def:rho_seg}. Consequently, \(p_{ij}(x,u_i)(h(x)-j\Delta)\geq0\) and \(-q_{i,j}(x,u_i)(h(x)-(j-1)\Delta)\geq0\). Similarly, for all \(u_i\in \mathcal{U}_i\), the input constraint inequalities imply $-\sum_{k=1}^{r_i}(w_{i,j,k}(x,u_i)(u_{i,k}-\underline{u}_{i,k})\geq 0$  and $\sum_{k=1}^{r_i} v_{i,j,k}(x,u_i)(\overline{u}_{i,k}-u_{i,k})\geq 0$. Thus, any $\rho_{ij}$ rendering constraint \eqref{eq:rho SOS 1} as an SOS satisfies $~\frac{\partial h}{\partial x_i} g_i(x) (u_i - \hat{u}_i(x)) - \rho_{ij}\geq 0$ or equivalently \eqref{eq:rho_ij} for $x \in \mathcal{C}_j$ and $u_i\in \mathcal{U}_i$. Since for quadratic polynomials the SOS condition is equivalent to nonnegativity ~\cite{powers2011positive}, 
we have $\rho_{ij} = \hat{\rho}_{ij}$ whenever the expression in \eqref{eq:rho SOS 1} is quadratic.
\end{proof}
Lemma~\ref{lemma:1} can be implemented offline to compute the segmented criticality indices $\{\rho_{ij}\}$. Once $\{\rho_{ij}\}$ have been computed for all $i=1,\cdots,N$ and $j=1, \cdots, K$. they can be used to assess whether a given mapping $M$ for architecture assignment ensures safety. Furthermore, if the cost or price associated with implementing each architecture is known, we can determine an optimal assignment that minimizes the total cost while guaranteeing the system's safety. Algorithm~\ref{algo:C1} presents an algorithm to compute the optimal architecture assignment $M_{opt}$, where $p(A_j)$ and $t(A_j)$ denote the implementation cost and recovery time of architecture $A_j$, respectively.

 \begin{center}
  	\begin{algorithm}[!htp]
  		\caption{Optimal Assignment of Resilient Architectures}
  		\label{algo:C1}
  		\begin{algorithmic}[1]
  		\State \textbf{Input:} Dynamics $f_i(x)$, $g_i(x)$, control bounds $\{\underline{u}_{i}\},\{\overline{u}_{i}\}$ for $i = 1, 2, \cdots, N$, safety function $h(x)$, parameter $c$, number of segments $K$, architecture set $\{A_1, A_2, \dots, A_J\}$.
        \State \textbf{Output:} Optimal mapping $M_{opt}$.
        \State Compute $\{\rho_{ij}\}$ using Lemma 1 for all $i=1, 2,  \cdots, N$ and $j=1, 2, \cdots, K$.
        \State Sort $\{\rho_{ij}\}$ in ascending order: $\{\rho'_{ij}\}\gets$ sort($\{\rho_{ij}\}$)
        \State  \textbf{Initialization:} $M_{opt}=\emptyset,~d=\infty$.
        \For{Each mapping $M: \mathcal{S} \rightarrow \mathcal{A}$}
            \State Check if \eqref{eq: new safety constraint} is satisfied.
            \If{Yes}
            \If{$d>\sum_{i=1}^N p(M(S_i))$}
            \State $d \gets \sum_{i=1}^N p(M(S_i))$, $M_{opt}\gets M$.
            \EndIf
            \EndIf
         \EndFor
         \If{$M_{opt}=\emptyset$}
            \State No solution.
        \Else
             \State \textbf{return} $M_{opt}$
        \EndIf
  		\end{algorithmic}
  	\end{algorithm}
  \end{center}

%% file: simulation.tex
\section{Case Study} \label{sec:simulation}

\begin{table}[b]
\centering
\caption{Segmented criticality indices, recovery times, and degradation measures computed for the subsystems.}
\begin{tabularx}{\linewidth}{|c|X|X|X|}
\hline
\textbf{SCI} & \textbf{$S_1$} & \textbf{$S_2$} & \textbf{$S_3$} \\
\hline
$\rho_{i1}$ & -172.2 & -172.2 & -12.89 \\
$\rho_{i2}$ & -212.57 & -212.57 & -11.92 \\
$\rho_{i3}$ & -238.42 & -238.44 & -16.2 \\
$\rho_{i4}$ & -259.52 & -264.97 & -18.54 \\
$\rho_{i5}$ & -297.2 & -298.27 & -20.26 \\
$\rho_{i6}$ & -354.3 & -289.8 & -23.16 \\
$\rho_{i7}$ & -374.51 & -374.51 & -22.89 \\
$\rho_{i8}$ & -408.70 & -394.69 & -24.4 \\
\hline
\textbf{Recovery ($T$)} & 0.009192 & 0.009192 & 0.100917 \\
\hline
$D_i(T)$ & -1.80 & -1.80 & -1.25 \\
\hline
\multicolumn{4}{|c|}{$c + \sum_{i=1}^N D_i(T) \geq 0 \quad \Rightarrow \quad 5 - (1.80+ 1.80 + 1.25) = 0.15 \geq 0$} \\
\hline
\end{tabularx}
\label{tab:CI_recovery}
\end{table}

 In this section, we illustrate the effectiveness of our proposed approach using a case study focused on temperature regulation in a circular building comprising \( N \) interconnected rooms~\cite{girard2015safety}. Our simulation setup is similar to the one used in~\cite{maruf2022compositional}. Here, each room $i = 1,\ldots,N$ has an associated temperature state \( x_i \) that evolves according to the following continuous-time dynamic model:
\begin{equation}
    \resizebox{.9\columnwidth}{!}{$
    \dot{x}_i = \frac{1}{\delta}\left(w(x_{i+1} + x_{i-1} - 2x_i) + y(T_e - x_i) + z(T_h - x_i)u_i\right)$}\nonumber
\end{equation} 
where \( T_e \) represents the external ambient temperature, \( T_h \) denotes the heater temperature, and \( u_i \) is the control input for room \( i \). The variables \( x_{i+1} \) and \( x_{i-1} \) denote the temperatures in adjacent rooms, and a circular configuration is implemented through boundary conditions \( x_0 = x_N \) and \( x_{N+1} = x_1 \).

For the demonstrations, we set $N=3$, treating each room as a subsystem. The external and heater temperatures are fixed at \( T_e = -1^\circ\mathrm{C} \) and \( T_h = 50^\circ\mathrm{C} \), respectively. Control inputs are constrained within the intervals \( \mathcal{U}_1, \mathcal{U}_2 \in [-2, 2] \) and \( \mathcal{U}_3 \in [0, 0.6] \). System parameters are selected as \( w = 0.45 \), \( y = 0.045 \), \( z = 0.09 \), and the scaling factor is set to \( \delta = 0.1 \). To ensure safety, we require that the state vector \( x \in \mathbb{R}^3 \) remains within a safety set \( \mathcal{C}  = \{ x : h(x) \geq 0 \} \subset \mathbb{R}^3 \) for all \( t \geq 0 \), where the function \( h(x) \) is defined as
\begin{equation}
h(x) = \left( \frac{1}{3} \sum_{i=1}^{3} x_i - 15 \right) \left( 20 - \frac{1}{3} \sum_{i=1}^{3} x_i \right)\nonumber
\end{equation}

\begin{figure}[htbp]
    \centering
    \includegraphics[width=0.95\linewidth]{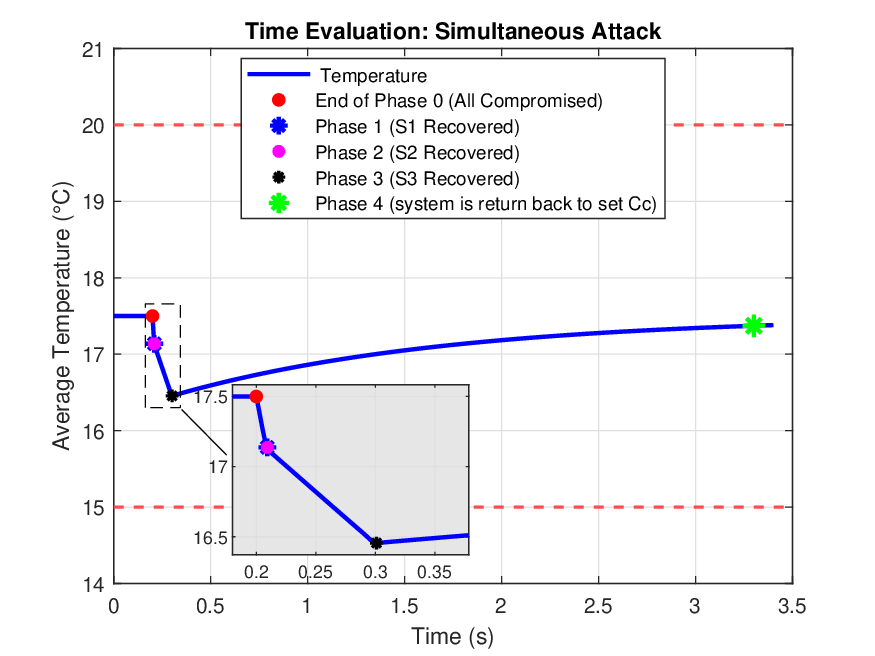}
    \caption[Time evolution: Simultaneous Attack.]{Time evolution of average temperature for simultaneous attack. All the systems are compromised simultaneously after Phase 0, which corresponds to time $t=0.2$ seconds. At Phase 1 (or equivalently, Phase 2), $S_1$ and $S_2$ are recovered together, as both have the same recovery time; however, $S_3$ remains compromised. At Phase 3, $S_3$ is also recovered. At Phase 4, the system returns to the set $\mathcal{C}_c$. As we see, the system maintains safety within the attack cycle.}
    \label{fig:full_attack_recovery}
\end{figure}

\begin{figure}[htbp]
    \centering
    \includegraphics[width=0.95\linewidth]{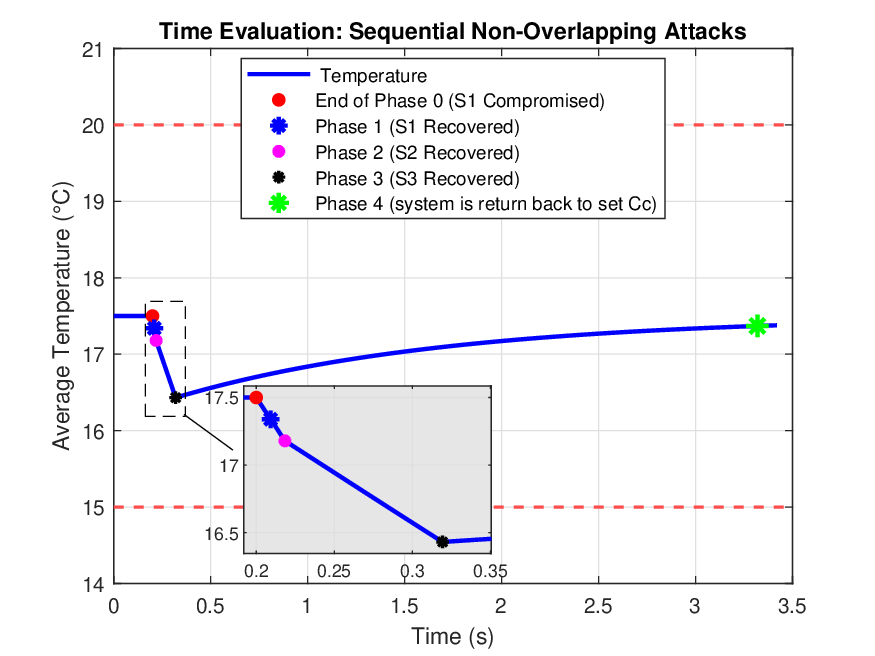}
    \caption[Time evaluation: Sequential Non-Overlapping Attacks.]{Time evolution of average temperature under sequential non-overlapping attacks. Subsystem $S_1$ is compromised after Phase~0 at  $t = 0.2$. During Phase 1, $S_1$ is recovered, and $S_2$ is compromised immediately. At Phase 2, $S_2$ is recovered, followed by the immediate compromise of $S_3$. At Phase 3, $S_3$ is recovered. Finally, at Phase 4, the system returns to the set $\mathcal{C}_c$. Throughout the attack cycle, the system maintains safety.}
    \label{fig:seq_attack_response}
\end{figure}

\begin{figure}[htbp]
    \centering
    \includegraphics[width=0.95\linewidth]{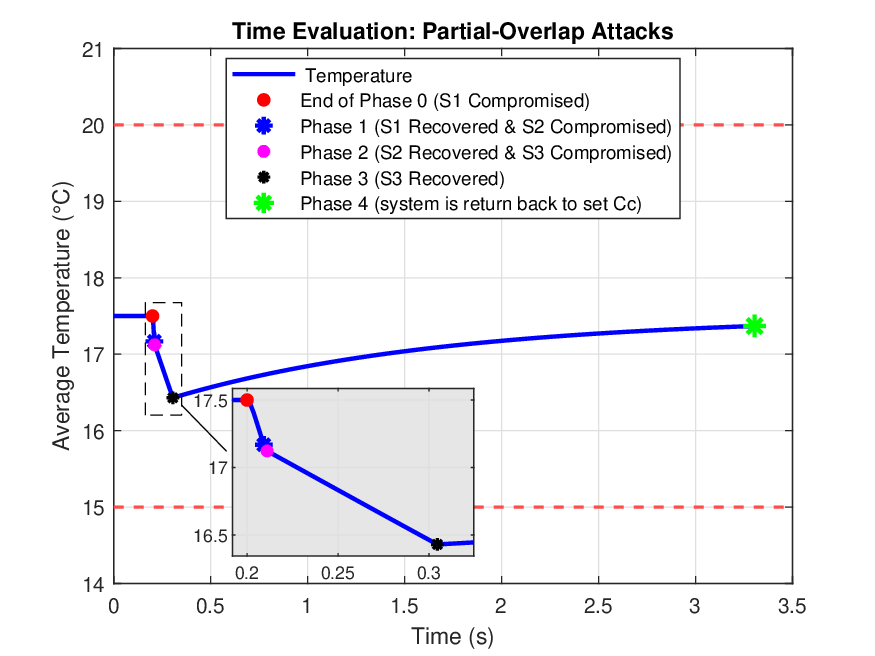}
    \caption[Time evolution: Partial Overlap Attack.]{
        Time evolution of average temperature under partial overlapping attack with an overlap interval of $\ t_{\mathrm{ov}} = 0.0018$~s. Subsystem $S_1$ is compromised after Phase~0 at  $t = 0.2$. During Phase 1, $S_1$ is recovered, and $S_2$ remains compromised. At Phase 2, $S_2$ is recovered and $S_3$ remains compromised. At Phase 3, $S_3$ is recovered. At Phase 4, the system returns to the set $\mathcal{C}_c$. Throughout the attack cycle, the system maintains safety. }
    \label{fig:overlap_attack}
\end{figure}

\begin{figure}[htbp]
    \centering
    \includegraphics[width=0.95\linewidth]{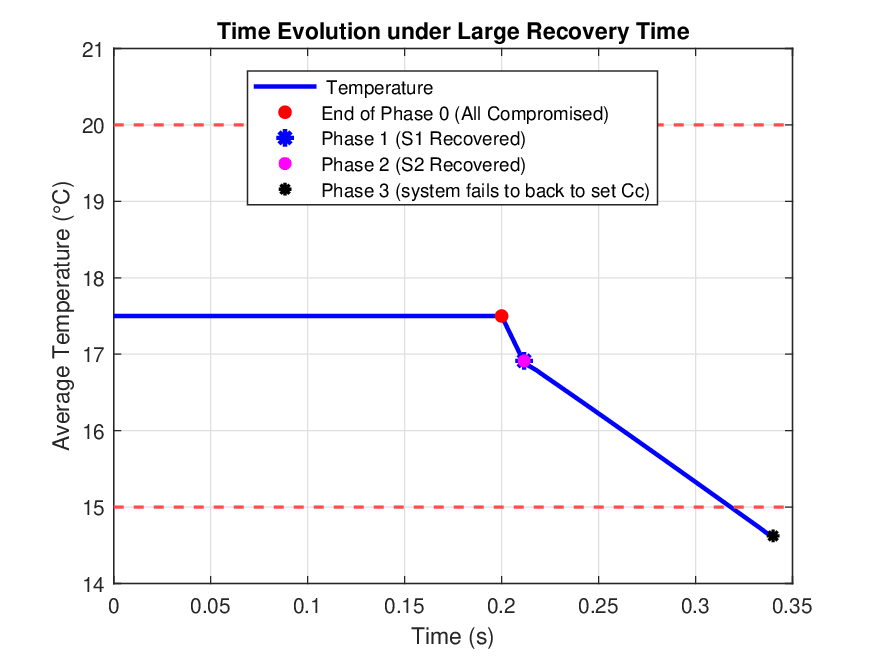}
    \caption [Time evolution: Simultaneous attacks with large recovery time ]{Time evolution of average temperature under simultaneous attacks with large recovery times. All subsystems are compromised simultaneously after Phase~0 at $t = 0.2$. The resilient architectures implemented in subsystems $S_1$, $S_2$, and $S_3$ exhibit recovery times of 0.0116 s, 0.0116 s, and 0.140 s, respectively. Due to its comparatively larger recovery time, $S_3$ fails to recover, causing the average temperature to fall below the threshold. Consequently, the system violates the safety constraint.}
    \label{fig:Large Recovery Time}
\end{figure}

This constraint ensures that the average temperature \( \frac{1}{3} \sum_{i=1}^{3} x_i \) remains within the range \([15, 20]\) at all times. In our simulation, we set the safety threshold parameter to \( c = 5 \). We adopt a nominal policy that satisfies Assumption~\ref{assump1} with the timing parameter \( \tau = 3 \) and is chosen as a linear state-feedback law.

Following our proposed approach, we first compute the CI values. We select $K=8$ and obtain \( \rho_{i1}, \ldots, \rho_{i8} \) for every subsystem \( S_1, S_2, S_3 \), as summarized in Table~\ref{tab:CI_recovery}. These indices reflect the potential degradation rates across the segments $\mathcal{C}_1, \cdots, \mathcal{C}_8$, as defined in Definition \ref{def:rho_seg}. 
Next, we select the recovery times for the subsystems to be $0.009192$ s, $0.009192$ s, and $0.100917$ s for \( S_1, S_2, S_3 \), respectively. This means that the implemented resilient architectures in subsystems \( S_1, S_2, S_3 \) have $0.009192$ s, $0.009192$ s, and $0.100917$ s of recovery time, respectively. We then determine the worst-case degradation \( D_i(\cdot) \) for each subsystem based on its selected architecture's recovery time. The final row shows that the total degradation across all subsystems does not exceed the safety margin. Therefore, the assigned architectures satisfy the conditions of Theorem~\ref{Th3}, ensuring that the overall system remains safe under any cyberattack scenario within an attack cycle.

Now we consider three distinct cyberattack scenarios in an attack cycle: (i) all subsystems are compromised simultaneously (similar to the scenario shown in Fig. 1 (a)), (ii) each subsystem is attacked sequentially without overlap among the compromised
subsystems (similar to Fig. 1 (b)), and (iii) subsystems are attacked with partial temporal overlap among the compromised subsystems (similar to Fig. 1 (c) with an overlap of $0.0018$ s). 
During a compromise, the attacker selects an input within the bound $u_i \in \mathcal{U}_i$ to drive the average temperature below the safety threshold. The time evolutions of average temperature under these three scenarios are demonstrated in Figs ~\ref{fig:full_attack_recovery}, \ref{fig:seq_attack_response}, and \ref{fig:overlap_attack}, respectively.  In all scenarios, the system remains within its safe region, as guaranteed by Theorem 3. Furthermore, we observe that once all nominal controllers are reinstated, the system returns to the set $\mathcal{C}_c=\{x:h(x)\geq c\}$ within $\tau=3$ seconds. Therefore, if a new attack cycle begins while the system is within $\mathcal{C}_c$ (as indicated by the green point highlighted in the figures), the system can withstand this subsequent attack cycle as well. Note that the minimum time gap required between consecutive attack cycles is dictated by $\tau$, which can be reduced by designing a new nominal policy.

For comparison, Fig.~\ref{fig:Large Recovery Time} shows the case of simultaneous attacks with large recovery times. In this scenario, the condition in Theorem~\ref{Th3} is not met, and the system leaves the safety set, highlighting the importance of selecting recovery times that satisfy \eqref{eq: new safety constraint}.

%% file: conclusion.tex
\section{Conclusion}\label{sec:conclusion}
\normalfont

In this paper, we studied the problem of assigning resilient architectures and designing their recovery times of CPS in an interconnected CPS with a safety guarantee against cyberattacks on actuators. In our approach, we constructed a scalar metric, referred to as criticality index (CI) for each subsystem and derived a sufficient condition for safety in the form of a linear inequality that links each subsystem’s CI to its recovery time. We then proposed a segmentation-based strategy to tighten the derived condition. Finally, we presented the computation of the indices as a sum-of-squares (SOS) optimization program and an algorithm to find a cost-optimal architecture assignment with the safety guarantee. The proposed solution was validated on a temperature-regulation case study in a three-room circular building for multiple attack scenarios. Our solution admits offline application to inform the design of interconnected CPS and the selection of recovery times for the employed cyber-resilient architectures. Future work will expand the framework to a broader range of cyberattacks, including attacks targeting sensors and communication links.